\documentclass[12pt]{article}
\usepackage[english]{babel}

\usepackage{amsmath,amsthm,amsfonts,amssymb,multicol,xcolor}
\usepackage{fullpage}
\usepackage{url}

\usepackage{verbatim} 
\usepackage{enumitem}
\usepackage{bbm} 
\usepackage[normalem]{ulem} 
\usepackage{bm}
\usepackage{mathtools}
\usepackage[pdftex]{hyperref}

\usepackage{dsfont}

\newcommand{\beq}{\begin{equation}}
\newcommand{\eeq}{\end{equation}}
\newcommand{\beqs}{\begin{equation*}}
\newcommand{\eeqs}{\end{equation*}}
\renewcommand{\l}{\left}
\renewcommand{\r}{\right}

\newcommand{\setof}[2]{\left\{ #1\; : \;#2 \right\}}

\newcommand{\ket}[1]{|#1\rangle}

\newcommand{\abs}[1]{\left\vert#1\right\vert}

\newcommand{\supp}{\,\text{supp}}

\newtheorem{thm}{Theorem}[section]
\newtheorem{lm}[thm]{Lemma}
\newtheorem{cor}[thm]{Corollary}
\newtheorem{prop}[thm]{Proposition}

\theoremstyle{definition}
\newtheorem{defn}[thm]{Definition}

\newtheorem{assum}[thm]{Assumption}

\theoremstyle{remark}
\newtheorem{rmk}[thm]{Remark}

\providecommand{\norm}[1]{\lVert#1\rVert}

\DeclareMathOperator{\gap}{gap}

\usepackage[doi=true, style=alphabetic, isbn=false, backend=bibtex]{biblatex}
\renewbibmacro{in:}{} 

\usepackage{authblk}

\makeatletter
\newcommand\blankfootnote[1]{%
  \begingroup
    \renewcommand\thefootnote{}%
    \footnote{#1}%
    \addtocounter{footnote}{-1}%
  \endgroup
}
\makeatother

\begin{document}

\title{On the critical finite-size gap scaling for frustration-free Hamiltonians}
\author[1]{Marius Lemm \thanks{\tt marius.lemm@uni-tuebingen.de}}
\affil[1]{
Department of Mathematics, University of Tübingen, 72076 Tübingen, Germany
}
\author[2,3,4]{Angelo Lucia \thanks{\tt angelo.lucia@polimi.it}}
\affil[2]{
Dipartimento di Matematica, Politecnico di Milano, 20131 Milano, Italy
}
\affil[3]{
    Departamento de Análisis Matemático y Matemática Aplicada, 
    Universidad Complutense de Madrid, 28040 Madrid, Spain
}
\affil[3]{
    Instituto de Ciencias Matemáticas,
    28049 Madrid, Spain
}
\date{June 7, 2025}

\maketitle

\abstract{We prove that the critical finite-size gap scaling for frustration-free Hamiltonians is of inverse-square type. The result covers  general graphs embedded in $\mathbb R^D$ and  general finite-range interactions without requiring assumptions about the ground state correlations. Therefore, the inverse-square critical gap scaling is a robust, universal property of finite-range frustration-free Hamiltonians. This places further limits on their ability to produce conformal field theories in the continuum limit. Our proof refines the divide-and-conquer strategy of Kastoryano and the second author through the refined Detectability Lemma of Gosset--Huang. }\blankfootnote{
}\blankfootnote{
}

\section{Introduction, setup and main result}
\subsection{Introduction}

A common strategy to prove lower bounds to the spectral gap for frustration-free quantum spin systems on a lattice is to obtain sufficiently good estimates on the spectral gap on finite volumes. Under certain assumptionss, these can in fact be lifted to uniform estimates for arbitrarily large volumes, which then implies a positive spectral gap in the thermodynamic limit. Both Knabe-type finite-size criteria (which require a single finite volume estimate)
\cite{gosset2016local,knabe1988energy,lemm2019spectral,lemm2020finite,lemm2022quantitatively} as well as the divide-and-conquer approach \cite{Kastoryano_2018} (which requires asymptotic estimates on the gap for large volumes) fit in this philosophy. 

The contrapositive of these results is that, in the case of gapless systems, the spectral gap of finite volumes (which is always a positive number) has to vanish sufficiently fast in the thermodynamic limit, as one could otherwise use the slowly vanishing finite volume bound to bootstrap a constant bound to the gap in the thermodynamic limit. One then can speak of a \emph{critical finite-size gap scaling} for frustration-free Hamiltonians: the slowest rate at which the spectral gap of any gapless frustration-free model vanishes. 

This idea was first observed by \cite{gosset2016local}, where a critical rate of $l^{-2}$ was proven for 1D spin chains with nearest neighbors interactions (where $l$ denotes the length of the chain), as well as for 2D nearest neighbor models on $l\times l$ square lattices (where $l$ is therefore now proportional to the diameter of the volume). The inverse square scaling $l^{-2}$ is well-known to be optimal in view of the Heisenberg ferromagnet, whose spectral gap vanishes as an inverse square on hypercubic lattices in any dimension, which can be easily proved rigorously by constructing explicit spin wave trial states.

This line of work was subsequently generalized in \cite{Anshu_2020,lemm2022quantitatively, masaoka2024rigorous}, where the inverse-square critical threshold was proven to hold for any interaction supported on the unit cells of hyper-cubic lattices \cite{Anshu_2020}, for nearest-neighbor interactions on the honeycomb and triangular lattices \cite{lemm2022quantitatively}, and for general graphs under the additional assumption that the ground state correlations decay as a power-law decaying \cite{masaoka2024rigorous}. Building up on a different proof technique, in \cite{Kastoryano_2018} it was shown a critical finite-size gap scaling of $l^{-1}$, up to logarithmic correction, which applies to general finite range interactions on any hypercubic lattice. Thus, while the threshold proven there is strictly not optimal, it indicates that one could hope for a general result that does not depend on the specific shape and geometry of the interactions. Also,  the technique in \cite{Anshu_2020} can be generalized to low degree lattices in $\mathbb R^D$ where edges only connect vertices at Euclidean distance $1$.

Our main contribution is to improve the threshold obtained in \cite{Kastoryano_2018} to $l^{-2}$, and to show that it applies to a large class of graphs and general interactions which includes all previous ones. Indeed, it covers all finite-range $k$-body interactions on bounded-degree graphs. While additional coarse-graining arguments could be applied to expand the scope of the prior works, they would not give a result of this generality. However,  Euclidean-type control on the graph geometry is still required in our proof and so the result does not apply, e.g., to expander graphs. The result comes with a mild logarithmic correction, which is unimportant for applications. 

Our result shows that the \textit{inverse-square gap closing speed is a universal property of finite-range frustration-free Hamiltonians}. In physics language, we prove that the dynamical critical exponent $z$ of finite-range frustration-free Hamiltonians satisfies $z\geq 2$ on general graphs and for general interactions. 
We achieve this generality by using the generality of the divide-and-conquer strategy and amplifying it to near-optimal threshold scaling $\sim l^{-2}$ up to logarithmic corrections via the refined Detectability Lemma due to Gosset-Huang \cite{gosset2016correlation}; see also \cite{anshu2016simple,Anshu_2020,anshu2020entanglement}. 

The approach is relatively flexible with respect to boundary conditions of $\Gamma$; see  Remark \ref{rmk:main}. In particular, it applies to open boundary conditions. Since it excludes gap closing rates that are of inverse length type $l^{-1}$, which would be characteristic of edge modes in conformal field theories \cite{francesco2012conformal,kitaev2006anyons,lemm2019spectral}, it places further limits on the ability of finite-range frustration-free Hamiltonians to produce conformal field theories in the continuum limit.

\subsection{Assumptions on the graph and the Hamiltonian}
We consider graphs $\Gamma$  that can be embedded in $\mathbb R^D$.

\begin{assum}[Assumption on the graph]\label{assum:bg}
    There exist a dimensional parameter $D\geq 1$, a constant $C_\Gamma\geq 1$, and an injective map $\iota:\Gamma\to\mathbb R^D$  such that
    \begin{equation}\label{eq:graph}
            C_\Gamma^{-1} |\iota(i)-\iota(j)|\leq d_\Gamma(i,j)\leq C_\Gamma |\iota(i)-\iota(j)|,\qquad \textnormal{for all }i,j\in\Gamma.
    \end{equation}
\end{assum}

Assumption \ref{assum:bg} is a rather weak assumption that allows many types of graph structures. The prime example satisfying Assumption \ref{assum:bg} is the hypercubic lattice $\mathbb Z^D$ and many standard lattices are also possible (e.g., honeycomb, triangular, fcc, bcc, hcp).

Assumption \ref{assum:bg} implies that the maximal degree of $\Gamma$ is bounded (see Lemma \ref{lm:basicPhi} for this fact and other consequences of Assumption \ref{assum:bg}).  
Examples where Assumption \ref{assum:bg} is \textit{not} satisfied include graphs of unbounded degree and the Bethe lattice (or other expander graphs). While $\Gamma$ is not required to have infinite cardinality, all relevant examples we have in mind have infinite cardinality.

Let $X\subset \Gamma$ be a finite subset. We place at each site of $X$ a $d$-dimensional local Hilbert space. Accordingly, the Hilbert space on the region $X$ is set to
\[
\mathcal H_X=\bigotimes_{i\in X} \mathbb C^d.
\] 
We write $\mathcal A_X$ for the associated local algebra of bounded linear operators with the natural inclusion $\mathcal A_X\subset \mathcal A_Y$ for $X\subset Y$ obtained by taking $A_X\otimes \mathrm{Id}_{Y\setminus X}\in \mathcal A_Y$.  To every finite $X$, we associate an interaction $\Phi(X)\in \mathcal A_X$, a Hermitian bounded linear operator that is subject to the following conditions.


\begin{assum}[Assumptions on the interaction]\label{assum:Phi}
~\begin{itemize}
\item[(i)] $\Phi$ has finite interaction range $R>0$, i.e., $\Phi(X)=0$ whenever $\mathrm{diam}\, (X)>R$. 

    \item[(ii)] $\Phi(X)\geq 0$ for all $X$ and it is uniformly bounded from above and away from zero, i.e.,
    \begin{align}
          \Phi_{\max}:=&\sup_{X\subset \Gamma}\|\Phi(X)\|<\infty,\\
          \Phi_{\min}:=&\inf_{\substack{X\subset \Gamma:\\ \Phi(X)\neq 0}}\gap\Phi(X)>0,
    \end{align}
    where we introduced $\gap\Phi(X):=\min(\mathrm{spec}\,\Phi(X)\setminus\{0\})$.
\end{itemize}
\end{assum}

Assumption (ii) is a uniformity assumption that allows us to reduce the Hamiltonian to a sum of local projectors. It holds in most of the practically relevant cases, e.g., in any translation-invariant situation. For finite-range interactions, Assumption (ii) is only violated in rather contrived situations, e.g., when the interactions behave very differently at infinity than in the bulk.





 Let $\Lambda\subset\Gamma$ be finite. On the Hilbert space $\mathcal H_\Lambda$, we consider the Hamiltonian
\begin{equation}\label{eq:HLambdadefn}
    H_\Lambda = 
\sum_{X\subset\Lambda} \Phi(X)
\end{equation}

Our main assumption is that the Hamiltonian is frustration-free. 

\begin{assum}[Assumption on the Hamiltonian]\label{assum:FF}
    For every finite $\Lambda\subset \Gamma$, $H_\Lambda$ is frustration-free, i.e., $\ker H_\Lambda \neq \{0\}$.
\end{assum}

The main quantity of interest is the  \emph{spectral gap} $\gap(\Lambda)$ of $H_\Lambda$, 
\begin{equation}
    \gap(H_{\Lambda})= \inf\left(\mathrm{spec}\, H(\Lambda)\setminus\{0\}\right)
\end{equation}
i.e., its smallest non-zero eigenvalue.

Well-known examples of frustration-free Hamiltonians that can be defined on general graphs $\Gamma$ are the Heisenberg ferromagnet,   AKLT Hamiltonians \cite{abdul2020class,affleck1988valence,lucia2023nonvanishing,wei2014hybrid}, and Hamiltonians with low-rank interactions \cite{sattath2016local,jauslin2022random,hunter-jones2025gapped}. We recall that almost all mathematical investigations of the size of spectral gaps require frustration-freeness --- without this assumption even absolutely fundamental questions like the Haldane conjecture \cite{haldane1983continuum,haldane1983nonlinear} (which asserts that any integer-spin Heisenberg antiferromagnetic chain is gapped) are completely open. 

Our contribution in this work is to show that \textit{ inverse-square critical gap scaling $l^{-2}$ is a robust, universal property of gapless frustration-free Hamiltonians}, in the sense that it holds under the rather weak Assumptions \ref{assum:bg} and \ref{assum:Phi} on the graph and interactions.



\subsection{Main result}

For simplicity, we identify from now on a graph $\Gamma$ satisfying Assumption \ref{assum:bg} with its Euclidean embedding $\iota(\Gamma)\subset \mathbb R^D$. This identification can be made without loss of generality, which can be seen by applying $\iota^{-1}$ to the Euclidean rectangles considered below.

 We use the following class of domains from \cite{Kastoryano_2018} based on a construction originally appearing in \cite{Cesi_2001}. For each $k\in \mathds{N}$, we let $l_k = (3/2)^{k/D}$ and introduce the  Euclidean rectangle 
\begin{equation}
        R(k) = [0, l_{k+1}] \times \cdots \times [0, l_{k+D}]\subset \mathbb R^D.
\end{equation}
We define $\mathcal{F}_k$ to be the collection of \textit{subsets of $\Gamma$ which are contained in $R(k)$} up to translations and permutations of the Euclidean coordinates.  We denote 
\begin{equation}
    \gap(\mathcal{F}_k) = \inf_{\Lambda\in \mathcal{F}_k}\gap(H_{\Lambda}), \quad \gap(\Gamma) = \inf_{k\geq 1} \gap(\mathcal{F}_k).
\end{equation}

\begin{thm}[Main result]\label{thm:main}
    Suppose that 
    \[
    \gap(\Gamma)=0.\]
   Then, for every $\epsilon>0$,
    \begin{equation}\label{eq:main}
        \gap(\mathcal{F}_k) = o\left( \frac{k^{4+\epsilon}}{l_k^2} \right), \qquad \textnormal{as } k\to\infty.
    \end{equation}
\end{thm}

That is, if the infinite-volume gap vanishes, then the finite-size gaps $\gap(\mathcal{F}_k)$ \textit{``cannot close too slowly''}: they must close at least as the inverse square power of the linear dimension, $l_k^{-2}$.

We emphasize that the numerator $k^{4+\epsilon}$ is a mild logarithmic correction, since the length parameter $l_k = (3/2)^{k/D}$ grows exponentially in $k$.\\

\begin{rmk}\label{rmk:main}
We comment on the special role played by periodic boundary conditions (b.c.). While our result is flexible with respect to the graph structure, the boundary conditions of $\Gamma$ are chosen globally in the beginning. This means that the subsystem Hamiltonians \eqref{eq:HLambdadefn} tend to have open b.c.    
    Our result therefore does not imply that periodic b.c.\ have inverse-square gap closing speed.     
    However, it is flexible enough to yield the following statement. Assume that the minimal gap taken over sets in $\mathcal F_k$ closes for all $2^d$ combinations of open b.c.\ and periodic b.c.\ that can be chosen for the rectangles $R(k)$. Then, this gap is in fact bounded from above by $o\left( \frac{k^{4+\epsilon}}{l_k^2} \right)$. In particular, the periodic gap is bounded by $o\left( \frac{k^{4+\epsilon}}{l_k^2} \right)$ and thus shows inverse-square scaling. Therefore, we can only conclude gap scaling of periodic b.c.\ if we also have information concerning gaps for open b.c.
Such a restriction is found throughout the literature, essentially because any iterative procedure that decomposes periodic b.c.\ will naturally encounter open b.c.\ as well. There is a simple workaround for this conundrum, if one is willing to assume that closing of the gap automatically implies the stronger statement that infinitely many eigenstates approach zero-energy (a statement that has indeed been proved for certain, very special 1D Hamiltonians), because then the result for periodic b.c.\ follows from the result for open b.c.\ from the min-max principle \cite{katsura2015exact,wouters2018exact,masaoka2024quadratic}. 
\end{rmk}

\section{Proof strategy}
\subsection{Reduction to sum of projectors}\label{ssect:projection}
In a preliminary step, we replace each $H_\Lambda$ by the sum of projectors
\begin{equation}\label{eq:projectionH}
   \tilde H_\Lambda=\sum_{X\subset \Lambda} h_X 
\end{equation}
where $h_X$ is the projection onto the range of $\Phi(X)$ (and so $h_X=0$ if $\Phi(X)=0$). The fact that $\Phi_{\min}h_X\leq \Phi(X)\leq \Phi_{\max} h_X$ implies that for any finite $\Lambda\subset \Gamma$,
\[
\Phi_{\min}\gap(\tilde H_{\Lambda})\leq 
\gap(H_\Lambda)\leq \Phi_{\max}\gap(\tilde H_{\Lambda}).
\]
Therefore, it suffices to prove Theorem \ref{thm:main} for $\tilde H_\Lambda$. Indeed, the assumption $\gap(H(\Lambda))=0$ implies $\gap(\tilde H_{\Lambda})=0$ and the conclusion that $\inf_{\Lambda\in \mathcal{F}_k}\gap(\tilde H_{\Lambda})=o\left( \frac{k^{4+\epsilon}}{l_k^2} \right)$ implies $\inf_{\Lambda\in \mathcal{F}_k}\gap(H_{\Lambda})=o\left( \frac{k^{4+\epsilon}}{l_k^2} \right)$.  In the following, we always consider the Hamiltonian \eqref{eq:projectionH}.

\subsection{Compatible sequences of domains}
The point of the choice of the collection of volumes $\mathcal{F}_k$ is that each element of $\mathcal{F}_k$ can be obtained as the union of two overlapping elements of $\mathcal{F}_{k-1}$. In fact, this decomposition is not unique, and it can be realized in many different ways. We will use a result from \cite{Kastoryano_2018}, with a minor improvement appearing in \cite{Lucia_2023}, where this fact is leveraged to bound $\gap(\mathcal{F}_k)$ in terms of $\gap(\mathcal{F}_{k-1})$.

As the result is more general, and only require certain properties of the sets $\mathcal{F}_k$, we introduce the following notation. 
Let $\mathcal{F}$ denote a family of finite subsets of $\Gamma$.
Denote
\begin{equation}
    \gap(\mathcal{F}) = \inf_{\Lambda \in \mathcal{F}} \gap (\Lambda),\qquad \gap(\Lambda)=\gap(\tilde H(\Lambda)).
\end{equation}
Given a finite $A\subset\Gamma$, we write $P_A$ for the projection onto the frustration-free ground state of $\tilde H_A$ (which is identical to the frustration-free ground state of the original Hamiltonian $ H_A$).

The following result shows how to get a relative bound on the spectral gap of two different families of finite subsets of $\Gamma$.
\begin{thm}[\protect{\cite[Theorem 2.3]{Lucia_2023}}]
\label{thm:gap-recursion}
    Let $\mathcal{F}$ and $\mathcal{F'}$ be two families of finite subsets of a graph $\Gamma$ satisfying Assumption \ref{assum:bg}. Suppose that there exists $s \in \mathds{N}$ and $ \delta \in [0,1]$ satisfying the following property: for each $Y \in \mathcal{F}'\setminus \mathcal{F}$, there exists $(A_i,B_i)_{i=1}^s$ pairs of elements in $\mathcal{F}$ such that:
    \begin{enumerate}
        \item $Y = A_i \cup B_i$ for each $i = 1, \dots , s$;
        \item $(A_i \cap B_i) \cap (A_j \cap B_j ) = \emptyset$ whenever $i\neq j$;
        \item $\| P_{A_i} P_{B_i} - P_Y \| \le \delta $ for every $i = 1, \dots , s$.
    \end{enumerate}
    Then 
    \begin{equation}
        \gap(\mathcal{F'}) \ge \frac{1-\delta}{1+\frac{1}{s}} \gap(\mathcal{F}).
    \end{equation}
\end{thm}

In \cite{Kastoryano_2018}, it was verified that the families $\mathcal{F}_k$ and $\mathcal{F}_{k-1}$ satisfy the assumptions of Theorem~\ref{thm:gap-recursion} for every $k$, with the appropriate choice of $s$ and $\delta$. In fact, while it was not explicitly stated, from the construction it can be seen that the decomposition has one additional property: the decomposition always happens along one coordinate direction, or in other words the sets $A_i$ and $B_i$ coincide when we project out one given direction (which depends on the choice of $Y$ but not of $i$). This will turn out to be important for our proof later, in order to connect to the refined Detectability Lemma, which requires us to make a 1D reduction to choose a coarse-graining direction. 

Let $\alpha\in \{1,\ldots,D\}$. We write $\Pi_\alpha$ for the coordinate projection
\[\Pi_\alpha(x_1,\ldots,x_D)=(x_1,\ldots,x_{\alpha-1},0,x_{\alpha+1},\ldots,x_D).
\]

\begin{prop}[\protect{\cite[Proposition 1]{Kastoryano_2018}}]\label{prop:fk-decomposition}
 Let $s_k \le \frac{1}{8}l_k$.  For each $k\geq 1$, the conditions 1 and 2 of Theorem~\ref{thm:gap-recursion} are satisfied for $\mathcal F'=\mathcal F_{k+1}$, $\mathcal F=\mathcal F_{k}$ and $s=s_k$. Moreover, for each $Y \in \mathcal{F}_k \setminus \mathcal{F}_{k-1}$, it is possible to choose the pairs $(A_i,B_i)_{i=1}^{s_k}$ in such a way that
\begin{equation}\label{eq:AiBidist}
        d_\Gamma(A_i\setminus B_i, B_i\setminus A_i) \ge C_\Gamma^{-1}\left(\frac{l_k}{8 s_k}\right).
    \end{equation}
    Moreover, there exists $\alpha\in \{1,\ldots,D\}$ such that    \begin{equation}\label{eq:projected}
  \Pi_\alpha A_i=\Pi_\alpha B_i.
    \end{equation}
\end{prop}
We include a brief summary of the proof.

\begin{proof}[Proof of Proposition \ref{prop:fk-decomposition}]
    Let $k\geq 2$ and consider a $Y \in \mathcal{F}_{k}\setminus \mathcal F_{k-1}$. Then up to a translation and a permutation of the coordinates, $Y$ coincides with $Y'\cap \Gamma$ for some $Y' \subset R(k)$. We will show how to decompose the Euclidean subset $Y'$, and the corresponding decomposition of $Y$ will follow by intersecting with $\Gamma$ and inverting the translation and permutation of the coordinates.

    For $i=1, \dots, s_k$, let
    \begin{align*}
        L_i &= [0,l_{k+1}] \times \cdots \times [0, l_{k+D-1}] \times \left[0, \frac{1}{2}l_{k+D} + 2 i \frac{l_k}{8s_k}\right]; \\
        R_i &= [0,l_{k+1}] \times \cdots \times [0, l_{k+D-1}] \times \left[\frac{1}{2}l_{k+D} + (2 i -1)\frac{l_k}{8s_k}, l_{k+D}\right].
    \end{align*}
    One can see that $L_i \cup R_i = R(k)$ for each $i$, that $\Pi_D L_i = \Pi_D R_i$, and that $(L_i \cap R_i) \cap (L_j \cap R_j) = \emptyset$ for $i\neq j$. Moreover, $L_i, R_i$ are contained in $R(k-1)$ up to a cyclic shift of the coordinates. Finally, with respect to the Euclidean distance, 
    \[ d(L_i \setminus R_i, R_i\setminus L_i) \ge \frac{l_k}{8 s_k} .\]

    We now define $A_i = Y' \cap L_i\cap \Gamma$ and $B_i = Y' \cap R_i\cap \Gamma$. From the properties of $L_i$, $R_i$, we see that $A_i$ and $B_i$ belong to $\mathcal{F}_{k-1}$, and that they must be non-empty (otherwise $Y$ would have belonged to $\mathcal{F}_{k-1}$). It then follows that they satisfy also conditions 1 and 2 of Theorem~\ref{thm:gap-recursion}, and that $\Pi_D A_i = \Pi_D B_i$. From \eqref{eq:graph}, it follows that
    \[
    \begin{aligned}
        d_\Gamma(A_i \setminus B_i, B_i \setminus A_i)
        &= d_\Gamma( Y' \cap \Gamma\cap (L_i\setminus R_i), Y' \cap \Gamma\cap (R_i\setminus L_i) )\\
        &\ge C_\Gamma^{-1} d(L_i \setminus R_i, R_i\setminus L_i).
    \end{aligned}\]
\end{proof}

Iterating Theorem~\ref{thm:gap-recursion} along the sequence $\{\mathcal{F}_k\}_{k\geq 1}$ yields the following gap estimate.

\begin{cor}\label{cor:iterated-gap-estimate}
Set
\begin{equation}\label{eq:deltakdefn}
    \delta_k := \sup_{(A,B) \in S_k } \| P_{A} P_{B} - P_{A \cup B} \|, 
\end{equation}
where
    \begin{equation}\label{eq:Skdefn}
        S_k := \left\{ (A,B) \in \mathcal{F}_{k-1}\times \mathcal{F}_{k-1}\, :\, A \cup B \in \mathcal{F}_k,\,  \eqref{eq:AiBidist} \textnormal{ and } 
        \eqref{eq:projected} \textnormal{ hold }\right\}.
    \end{equation}
Suppose there exists a positive $k_0$ such that $\delta_k < 1$ for every $k \ge k_0$.
    Then,
    \begin{equation}\label{eq:gap-estimate}
        \gap(\Gamma) \ge \Phi_{\min}\gap(\mathcal{F}_{k_0}) \prod_{k = k_0}^\infty \frac{1-\delta_k}{1+\frac{1}{s_k}}.
    \end{equation}
\end{cor}
Note that the constant on the r.h.s.\ of \eqref{eq:gap-estimate} is strictly positive if and only if $(\delta_k)_k$ and $\left(\frac{1}{s_k}\right)_k$ are summable sequences (a stronger condition than the fact that $\delta_k$ is eventually smaller than 1). In this case, the ``infinite-volume'' gap  $\inf_{n\geq 1} \gap(\mathcal{F}_{n})$ is lower-bounded by a constant and one commonly says that the infinite system is gapped. 

The proof of Theorem \ref{thm:main} relies on the following bound on $\delta_k$ in terms of s lower bound on the spectral gap $\gap(\mathcal{F}_k)$.
\begin{lm}[Refined overlap bound]\label{lm:deltakbound}
    There exist $C_1,C_2>0$ such that the following holds. Let $k\geq 1$ and $(A,B)\in S_k$. Suppose that $\lambda_k\in (0,1)$ satisfies $\gap(\mathcal F_k)\geq \lambda_k>0$. 

    Then, 
    \begin{equation}\label{eq:deltakbound}
             \delta_k \le C_1 \exp\left(-C_2 \sqrt{\lambda_k} \frac{l_k}{s_k}\right).
\end{equation}
\end{lm}

This is the key improvement over Theorem 11 in \cite{Kastoryano_2018}, in which the exponent $\sqrt{\lambda_k}$ behaved as $\lambda_k$ instead.
We postpone the proof of Lemma \ref{lm:deltakbound} to Section~\ref{sec:overlap-bound} for the moment, and we first show how Theorem~\ref{thm:main} follows from it.

\subsection{Proof of Theorem \ref{thm:main} assuming Lemma \ref{lm:deltakbound}}

As a first step, we combine Lemma~\ref{lm:deltakbound} with Corollary~\ref{cor:iterated-gap-estimate} to obtain the following Proposition, showing how a lower bound to $\gap(\mathcal{F}_k)$ decaying sufficiently slowly (see \eqref{eq:decay-condition}) implies a positive gap in the thermodynamic limit.

\begin{prop}\label{prop:threshold-gap}
    Suppose there exists a sequence $\lambda_k \in (0,1) $ such that
    \begin{equation}\label{eq:lambdaklb}
         \gap(\mathcal{F}_k) \ge \lambda_k > 0, \quad \forall k;
    \end{equation}
    and a sequence $s_k > 0$ such that $(\frac{1}{s_k})_k$ is summable and
    \begin{equation}\label{eq:decay-condition}
        \liminf_{k\to\infty} \sqrt{\lambda_k}  \frac{l_k}{k s_k} > 0.
    \end{equation}
    Then $\gap(\Gamma) > 0$.
\end{prop}
\begin{proof}
    By Lemma \ref{lm:deltakbound} and the definition of $\delta_k$, we have that
    \begin{equation}\label{eq:deltakestimate}
           \delta_k \le C_1 \exp\left(-C_2 \sqrt{\lambda_k} \frac{l_k}{s_k}\right),
    \end{equation}
    for positive constants $C_1$ and $C_2$.
    In order to apply Corollary~\ref{cor:iterated-gap-estimate}, and to show that the r.h.s.\ of \eqref{eq:gap-estimate} is non-zero, we need to show that $(\delta_k)_k$ is a summable sequence. We will do so by checking that the r.h.s.\ of \eqref{eq:deltakestimate} satisfies the root test:
    \[
    \limsup_{k\to\infty}\ C_1^{1/k} \exp\left( - C_2 \sqrt{\lambda_k} \frac{l_k}{s_k}\right)^{1/k} =
    \exp\left(-C_2  \liminf_{k\to\infty} \sqrt{\lambda_k}  \frac{l_k}{k s_k}\right)
    \]
    which is strictly smaller than $1$ by assumption.
\end{proof}

We are now ready to prove Theorem \ref{thm:main}. It will follow by contradiction from Proposition~\ref{prop:threshold-gap}: if the system is gapless in the thermodynamic limit, then the assumptions of Proposition~\ref{prop:threshold-gap} cannot hold true, and therefore the finite-size volume gap has to vanish sufficiently fast.

\begin{proof}[Proof of Theorem \ref{thm:main}]
    As explained in Subsection \ref{ssect:projection}, it suffices to prove the claim for the Hamiltonians 
    \eqref{eq:projectionH}. We reason by contradiction. Suppose that
    \[
    \gap(\mathcal{F}_k) \ge c \frac{k^{4+\epsilon}}{l_k^2}, \quad \forall k,
    \]
    for some positive constant $c$ and some $\epsilon >0$. Set $ \lambda_k=c \frac{k^{4+\epsilon}}{l_k^2}$, which is $<1$ for sufficiently large $k$, and choose $s_k = k^{1+\frac{\epsilon}{2}}$, making $(\frac{1}{s_k})_k$ summable. Moreover
    \[
    \liminf_{k\to\infty} \sqrt{\lambda_k}  \frac{l_k}{k s_k} = \sqrt{c} \liminf_{k\to\infty} \frac {k^{2+\frac{\epsilon}{2}}}{l_k} \frac{l_k}{k s_k} = \sqrt{c}>0.
    \]
    By Proposition~\ref{prop:threshold-gap}, this implies that $\gap(\Gamma) > 0$, leading to a contradiction.
\end{proof}


\section{Refined overlap bound}
\label{sec:overlap-bound}

In this section, we prove Lemma \ref{lm:deltakbound}. The proof will follow from a suitable generalization of the refined Detectability Lemma (Lemma \ref{lm:AAG} below), whose setup we describe next. In contrast to \cite{Kastoryano_2018} and rather similarly to \cite{Anshu_2020}, we first perform a 1D coarse-graining which is a key ingredient in the refined Detectability Lemma. The  differences compared to \cite{Anshu_2020} are that we work in a more general setup (general number of layers and interaction range) and we need to track how Euclidean versus graph geometry enters in the arguments, as both are relevant.

\subsection{Coarse-graining}\label{ssect:prelim}
Similarly to \cite{Anshu_2020}, as a first step we perform a coarse-graining of the interactions. Once a coordinate direction is fixed (which will be determined by the decomposition from Proposition~\ref{prop:fk-decomposition}), we will group the interactions into ``columns'' along the specified direction, in order to reduce the problem to a quasi-1D model.

Let $k\geq 1$ and consider $(A,B)\in S_k$. We denote $\Lambda=A\cup B$. Let $\alpha\in\{1,\ldots,D\}$ be such that $\Pi_\alpha A=\Pi_\alpha B$. Without loss of generality, by rotation, we can assume that $\alpha=1$.

We fix a coarse-graining parameter $t$ to be an integer such that
\[
t > \{2, C_{\Gamma} R\}
\]
and we define the width-$4t$ ``columns'' centered at $j\in \mathbb Z$,
\[
\mathcal C_j=([j-2t,j+2t]\times \mathbb R^{D-1})\cap \Lambda,\qquad j\in \mathbb Z
\]
as well as the associated coarse-grained Hamiltonians
\[
\tilde H_{j}=\sum_{X\subset \mathcal C_j} h_X.
\]
We write $Q_j$ for the ground state projector of $\tilde H_j$. 

In order to cover every interaction, we will consider two sets of ``columns'' (denoted as ``even'' and ``odd''), chosen in such a way that: columns belonging to the same set are disjoint (Lemma~\ref{lm:commuting}); columns from different sets which do intersect have a ``large'' (proportional to the coarse-graining parameter $t$) overlap.
\begin{defn}
    Let us define two sets of indices:
    \begin{equation}
    \mathcal{I}_e(t) := \{ (2+6j)\, t \mid j \in \mathds{Z} \},
    \quad 
    \mathcal{I}_o(t) := \{ (5+6j)\, t \mid j \in \mathds{Z} \}.
    \end{equation}
\end{defn}


\begin{lm}\label{lm:commuting}
  Let $\#\in \{e,o\}$ and let $m,n\in \mathcal I_\#(t)$ be distinct. Then $[Q_m,Q_{n}]=0$.
\end{lm}
\begin{proof}[Proof of Lemma \ref{lm:commuting}]
    We start by observing that, for any distinct pair of indices $m,n\in \mathcal I_\#(t)$, the Euclidean distance between the columns $\mathcal C_m$ and $\mathcal C_n$ is at least $2t$.
   In fact, when $\#=e$, $Q_m$ is supported on $[6jt,6jt + 4t] \times \mathbb{R}^{D-1}$ and $Q_{n}$ is supported on $[6j't,6j't+4t]\times \mathbb{R}^{D-1}$ for two distinct $j,j'\in\mathbb Z$. The smallest distance possible between the two intervals is $2t$ (which occurs when $\abs{j-j'}=1$). For $\#=o$, the argument is analogous.
   
    By \eqref{eq:graph}, the $\Gamma$-distance between $\mathcal C_n$ and $\mathcal C_m$ is at least $C_\Gamma^{-1}2t$. Since $t> \max\{2, C_{\Gamma} R\}$, we have 
   $C_\Gamma^{-1}2t > 2R$. Now $[Q_m,Q_{n}]=0$ follows because the interaction range is $R$. 
\end{proof}

Lemma \ref{lm:commuting} allows us to group the $Q_m$ operators into 
the even and the odd ones, which commute amongst each other.


\begin{defn}
The \emph{$t$-coarse grained detectability lemma operator} $DL(t)$ is given by
\begin{equation}
    DL(t) = \prod_{m \in \mathcal{I}_e(t)} Q_{m} \cdot \prod_{m\in \mathcal{I}_o(t)} Q_{m}. 
\end{equation}
\end{defn}

\subsection{Basic properties of the interaction}
We note two basic properties of the interaction $\Phi$ that follow from Assumption \ref{assum:Phi} and that are used in the Detectability Lemma.
\begin{lm}\label{lm:basicPhi}
There exist integers $g,L\geq 1$ such that the following holds.
    \begin{itemize}
        \item 
   $\Phi$ can be decomposed into $L$ layers, i.e., there exists a disjoint decomposition (``$L$-coloring'') of
\begin{equation}\label{eq:gcoloring}
    \setof{X\subset \Gamma}{\mathrm{diam}\, (X)\leq R}
=\mathcal X_1\sqcup \mathcal X_2\sqcup\ldots\sqcup \mathcal X_L
\end{equation}
such that for every $X,Y\in\mathcal X_i$, it holds that $X\cap Y=\emptyset$.
\item For every subset $X\subset \Gamma$ such that $\Phi(X)\neq 0$, there exist at most $g$ subsets $Y\neq X$ such that $[\Phi(X),\Phi(Y)]\neq 0$.
 \end{itemize}
\end{lm}

\begin{proof}[Proof of Lemma \ref{lm:basicPhi}]
    We first observe that Assumption \ref{assum:bg} implies that, for any $r>0$ and $i\in \Gamma$, the ball
    \[
    B^\Gamma_{r}(i):=\setof{j\in \Gamma}{d_\Gamma(i,j)\leq r}
    \]
    has bounded cardinality uniformly in $i$. This follows
    from the fact that $\iota(B^\Gamma_{r}(i))\subset \mathbb R^D$ is contained in the Euclidean ball $B_{C_\Gamma r}(i)$ and all its distinct points are at distance at least $C_\Gamma^{-1}>0$.
    
    Concerning the first bullet point, we 
    consider the hypergraph with vertices equal to $\Gamma$ and whose hyperedges are those $X\subset\Gamma$ for which $\Phi(X)\neq 0$. We write $\Delta(\Gamma)$ for the degree of the hypergraph, defined as usual as the maximal number of hyperedges intersecting at a vertex. By the finite-range assumption, given any $i\in \Gamma$, the number of $X$ containing $i$ is bounded by $2^{|B_{R}(i)|}$ and therefore this number also bounds the degree. As noted above, the bound on $|B_{R}(i)|$ is uniform in $i$ and so $\Delta(\Gamma)$ is bounded. 
        In this language, $L$ is called the edge-chromatic index of the hypergraph and a simple bound from graph theory is $L\leq \left\lfloor\frac{3\Delta(\Gamma)}{2}\right\rfloor$.

    Concerning the second bullet point, we note that $[\Phi(X),\Phi(Y)]\neq 0$ implies $X\cap Y\neq \emptyset$. By the finite-range assumption, $X$ and $Y$ must therefore lie in a $B_{2R}(i)$ for some $i\in \Gamma$. Hence, $g\leq 2^{|B_{2R}(i)|}-1$. As noted above, the bound on $|B_{R}(i)|$ is uniform in $i$ and so $g$ is bounded.
    \end{proof}

Of course, the bounds on $L$ and $g$ given in the above proof are far from optimal and much better bounds can usually be read off more simply. Since $L$ and $g$ do not enter in the main results, we do not dwell on this issue and refer readers interested in tighter bounds to the combinatorics literature, e.g., \cite{pippenger1989asymptotic} and the recent survey \cite{kang2023graph}.

\subsection{Refined Detectability Lemma}

The following refined Detectability Lemma is due to Gosset and Huang \cite{gosset2016correlation}; see also \cite{anshu2016simple,anshu2020entanglement}.

\begin{lm}[Refined Detectability Lemma]\label{lm:AAG} Let $k\geq 1$, $(A,B)\in S_k$ and $\gap(\mathcal F_k)\geq \lambda_k>0$. 
Denote $\Lambda=A\cup B$. For a normalized state $\ket{\phi}$ orthogonal to the ground state sector of $H_\Lambda$, it holds that
\begin{equation}
    \norm{DL(t) \ket{\phi}}\le   2 \exp\left(-\left(\frac{t}{C_\Gamma (L-1) R}-\frac{2L-1}{L-1}\right)\sqrt{\frac{\lambda_k}{1+g^2}} \right)
\end{equation}
\end{lm}

We now prove Lemma \ref{lm:AAG}.
The first step is to show that one can ``smuggle'' a polynomial $F$ with $F(0)=1$ and of controlled degree between the products defining $DL(t)$. This polynomial will later be chosen as a suitable Chebyshev polynomial.

Recall the notation \eqref{eq:gcoloring} for the $L$ layers. The \emph{standard detectability lemma operator} for $\tilde H_\Lambda$ from \eqref{eq:projectionH} is defined as
\begin{equation}
   T:=T_L\ldots T_1,\qquad T_\beta: = \prod_{X\in \mathcal X_\beta\cap \Lambda} (1-h_X),\quad \beta\in\{1,\ldots,L\}.
    \end{equation}

    The following lemma generalizes Lemma 3.1 in \cite{anshu2020entanglement} to any number of layers, interaction range, and Euclidean-embedded graphs.

\begin{lm}\label{lm:poly}
Let $F$ be a polynomial of degree at most $\left\lceil\frac{1}{2(L-1)} \left(\frac{t}{C_\Gamma R} -2L+1 \right)\right\rceil$, such that $F(0)=1$. Then
\begin{equation}\label{eq:F}
    DL(t) = \prod_{m \in \mathcal{I}_e(t)} Q_{m} \cdot F\left(\mathds{1} - T^\dag T\right)  \cdot \prod_{m\in \mathcal{I}_o(t)} Q_{m}.
\end{equation}
\end{lm}

Figure 5 in \cite{anshu2020entanglement} is an instructive depiction of the proof idea in one dimension.

\begin{proof}
For $\beta\in\{1,\ldots,L\}$ and $S \subset \Lambda$, define 
\[
T_\beta^S = \prod_{\substack{X\in\mathcal X_\beta\\ X\cap S\neq \emptyset}} (1-h_X).
\]
Let 
\[ 
S_0 = \Lambda \setminus \supp\l(\prod_{m \in \mathcal{I}_e(t)} Q_{m} \r)
\]
and define iteratively for $j\ge 1$,
\[
    S_j = \setof{ i \in \Lambda}{d_\Gamma(i,S_{j-1}) \le R }.
\]
Note that, by construction, $S_0 \subset S_1 \subset \cdots \subset S_n$ for a suitable $n\geq 1$.

Because of frustration freeness, we have that
\[
\prod_{m \in \mathcal{I}_e(t)} Q_{m} \cdot T_1 = \prod_{m \in \mathcal{I}_e(t)} Q_{m} \cdot T_1^{S_1}.
\]
Since each $T_\beta$ is a self-adjoint  projection, we have that 
\[
T^\dag T= T_1T_2\ldots T_{L-1}T_L T_{L-1}\ldots T_1,
\]
and so $(T^\dag T)^q = (T_1T_2\ldots T_{L-1}T_LT_{L-1}\ldots T_2)^q  T_1$. Therefore
\[
\prod_{m \in \mathcal{I}_e(t)} Q_{m} \cdot (T^\dag T)^q = \prod_{m \in \mathcal{I}_e(t)} Q_{m} \cdot T_1^{S_1} T_2^{S_2} \cdots  T_{2}^{S_{q(2L-2)}} T_1^{S_{q(2L-2)+1}}. 
\]
Now, if $S_{q(2L-2)+1}$ is contained in the support of $\prod_{m\in \mathcal{I}_o(t)} Q_{m}$, we have that
\[
T_1^{S_0} T_2^{S_1} \cdots  T_{2}^{S_{L(2L-2)}} T_1^{S_{q(2L-2)+1}}\cdot \prod_{m\in \mathcal{I}_o(t)} Q_{m} = \prod_{m\in \mathcal{I}_o(t)} Q_{m}.
\]
By construction, the set $S_0$ is at Euclidean distance $t$ from the complement of the support of $\prod_{m\in \mathcal{I}_o(t)} Q_{m}$ and so its graph distance to the complement of the support of $\prod_{m\in \mathcal{I}_o(t)} Q_{m}$ is at least $C_\Gamma^{-1} t$. By induction,  $S_j$ is at graph distance $C_\Gamma^{-1}t-jR$ from the complement of the support of $\prod_{m\in \mathcal{I}_o(t)} Q_{m}$. Therefore, this graph distance is non-zero as long as $j<t/(C_\Gamma R)$. Thus by choosing $q \leq\left\lceil\frac{1}{2(L-1)} \left(\frac{t}{C_\Gamma R} -2L+1 \right)\right\rceil$, we can show that
\[
\prod_{m \in \mathcal{I}_e(t)} Q_{m} \cdot (T^\dag T)^q \cdot \prod_{m\in \mathcal{I}_o(t)} Q_{m} = DL(t).
\]
By linearity, we obtain the stated claim.
\end{proof}

The identity \eqref{eq:F} is used to produce the polynomial and the outside products are then dropped, as summarized in the next corollary. Recall that $P_\Lambda$ denotes the projection onto the ground state of $\tilde H_\Lambda$. 

\begin{cor}\label{cor:Pperp}
We have
\begin{equation}
\norm{DL(t) P_\Lambda^\perp} \le \inf_{F} F_{*}(\norm{TP_\Lambda^\perp}^2),
\end{equation}
where the infimum runs over polynomials of degree at most $\left\lceil\frac{1}{2(L-1)} \left(\frac{t}{C_\Gamma R} -2L+1 \right)\right\rceil$ satisfying $F(0)=1$, and we set
\begin{equation}
    F_{*}(\epsilon) := \sup_{1-\epsilon \le x \le 1 }  \abs{F(x)}, \quad \forall \epsilon \in [0, 1].
\end{equation}
\end{cor}
\begin{proof}[Proof of Corollary \ref{cor:Pperp}]
Let $F$ be any of the admissible polynomials. By Lemma \ref{lm:poly},
\[ \norm{DL(t) P_\Lambda^\perp} \le \norm{F(\mathds{1} - T^\dag T)P_\Lambda^\perp}.\]
Then
\[ (1- \norm{TP_\Lambda^\perp}^2)\, \mathds{1} \le P_\Lambda^\perp(\mathds{1} - T^\dag T)P_\Lambda^\perp \le \mathds{1} \]
yields Corollary \ref{cor:Pperp}.
\end{proof}

We recall from \cite{arad2013area} a basic fact about Chebyshev polynomials. We let $\mathcal T_q$ denote the degree-$q$ Chebyshev polynomial.

\begin{lm}[\cite{arad2013area}]\label{lm:step}
For every integer $q\geq 1$ and number $\gamma\in (0,1)$, let
\[
\mathrm{Step}_{q,\gamma}(x):=\frac{\mathcal T_q\left(\frac{2(1-x)}{1-\gamma}-1\right)}{\mathcal T_q\left(\frac{2}{1-\gamma}-1\right)},\qquad x\in\mathbb R.
\]
Then
\begin{equation}
    |\mathrm{Step}_{q,\gamma}(x)|\leq 2\exp\left(-2q\sqrt{\gamma}\right) ,\qquad \gamma\leq x\leq 1.
\end{equation}
\end{lm}

\begin{proof}[Proof of Lemma \ref{lm:AAG}]
    Take any normalized vector $\ket{\psi^\perp}$ orthogonal to the ground state sector of $H_\Lambda$. By the standard Detectability Lemma \cite[Corollary 1]{anshu2016simple}, we have
    \[
    \|T\ket{\psi^\perp}\|^2\leq \frac{1}{1+\lambda_k/g^2}
    \]
    and therefore $\|TP_\Lambda^\perp\|^2\leq \frac{1}{1+\lambda_k/g^2}$.

    By applying Corollary \ref{cor:Pperp} and Lemma \ref{lm:step} with $q= \left\lceil\frac{1}{2(L-1)} \left(\frac{t}{C_\Gamma R} -2L+1 \right)\right\rceil$ and $\gamma=\frac{\lambda_k}{\lambda_k+g^2}\in (0,1)$, we obtain
    \[
    \begin{aligned}
            \|DL(t) \ket{\phi}\|\leq \inf_F F_*(\| T P_\Lambda^\perp\|^2)
    \leq& 2 \exp\left(-2\left\lceil\frac{1}{2(L-1)} \left(\frac{t}{C_\Gamma R} -2L+1 \right)\right\rceil\sqrt{\frac{\lambda_k}{\lambda_k+g^2}} \right)\\
    \leq& 2 \exp\left(-\left(\frac{t}{C_\Gamma (L-1) R}-\frac{2L-1}{L-1}\right)\sqrt{\frac{\lambda_k}{1+g^2}} \right)
    \end{aligned}
    \]
    where we used $\lambda_k\leq 1$ in the last step. This proves Lemma \ref{lm:AAG}.
\end{proof}

\subsection{Proof of Lemma \ref{lm:deltakbound}}
Let $k\geq 1$, $(A,B)\in S_k$, and denote $\Lambda=A\cup B$. We perform the coarse-graining as described in Subsection \ref{ssect:prelim}.

We now define $M_A$ and $M_B$ by reordering the product of $Q_m$ appearing in $DL(t)$ in such a way that $DL(t)=M_AM_B$. Let $\mathcal{I}_{e,A\setminus B}(t)$ be the set
\[
\mathcal{I}_{e,B\setminus A}(t) = \{ m\in \mathcal I_e(t) \mid \mathrm{supp}\, Q_m\cap (B\setminus A) \neq \emptyset \},
\]
i.e. $\mathcal{I}_{e,B\setminus A}(t)$ is the set of indices $m$ such that the support of $Q_m$ intersects $B\setminus A$. We then consider every operator $Q_{m'}$, with $m'\in \mathcal{I}_o(t)$, whose support intersects the support of any operator $Q_m$ for $m\in \mathcal{I}_{e,B\setminus A}(t)$:
\[
\mathcal{I}_{o,B\setminus A}(t) = \{ m' \in \mathcal I_o(t) \mid \exists m \in \mathcal{I}_{e,B\setminus A}(t) \text{ s.t. } Q_m\cap \mathrm{supp}\, Q_{m'}\neq \emptyset \}.
\]
By construction, if $m \in \mathcal{I}_{e,B\setminus A}(t)$ and $m' \in \mathcal I_o(t) \setminus \mathcal{I}_{o,B\setminus A}(t)$, the $Q_m$ and $Q_{m'}$ have disjoint support, and therefore they commute. Thus we can write
\[
DL(t) = \!\!\! \prod_{m \in \mathcal{I}_e(t) \setminus \mathcal{I}_{e,B\setminus A}(t)} Q_m \prod_{m' \in \mathcal I_o(t) \setminus \mathcal{I}_{o,B\setminus A}(t)} Q_m'
\prod_{m\in \mathcal{I}_{e,B\setminus A}(t)} Q_m  \prod_{m' \in  \mathcal{I}_{o,B\setminus A}(t)} Q_m' 
\]
We set
\begin{align*}
M_A &= \prod_{m \in \mathcal{I}_e(t) \setminus \mathcal{I}_{e,B\setminus A}(t)} Q_m  \prod_{m' \in \mathcal I_o(t) \setminus \mathcal{I}_{o,B\setminus A}(t)} Q_m' \\
M_B &= \prod_{m\in \mathcal{I}_{e,B\setminus A}(t)} Q_m  \prod_{m' \in \mathcal{I}_{o,B\setminus A}(t)} Q_m',
\end{align*}
so that $DL(t)=M_A M_B$. If the Euclidean distance between $A\setminus B$ and  $B\setminus A$ exceeds $8t$, then $M_A$ is supported on $A$ and $M_B$ is supported on $B$.

Recall that $P_A$ denotes the projection onto the ground state of $\tilde H_A$. We
have that $P_A = P_A M_A$, and therefore $P_A M_B = P_A M_A M_B = P_A DL(t)$. This implies that
\[
    \norm{ (P_A - M_A) M_B } = \norm{DL(t)P_A^\perp } \le \norm{DL(t) P_{AB}^\perp },
\]
where the inequality follows from the fact that $P_A^\perp \le P_{AB}^\perp$, a consequence of frustration-freeness.
Since it also holds that $P_B = P_B M_B$ and $P_B M_A = P_B DL(t)$, the analogous bound with $A$ and $B$ interchanged. This gives that
\[
\begin{aligned}
    \norm{P_A P_B - P_{AB}} \le& \norm{P_A P_B - M_A P_B}
    + \norm{M_A P_B - M_A M_B} + \norm{M_A M_B - P_{AB}} \\ 
    \le& \norm{(P_A - M_A)P_B} + \norm{M_A (P_B - M_B)} + \norm{DL(t)P_{AB}^\perp} 
    \le 3 \norm{DL(t)P_{AB}^\perp}.
\end{aligned}
\]
By Lemma \ref{lm:AAG}, we have 
\[
\norm{DL(t)P_{AB}^\perp}\leq 2 \exp\left(-\left(\frac{t}{C_\Gamma (L-1) R}-\frac{2L-1}{L-1}\right)\sqrt{\frac{\lambda_k}{1+g^2}} \right).
\]
By \eqref{eq:AiBidist} and Assumption \ref{assum:bg}, we can choose $t$ as
large as possible, as long as $8t < C_\Gamma^{-1} \frac{l_k}{8s_k}$.
We now claim that, from the last bound, it follows that for any $C_2 >0$ there exists a $C_1' > 0$ and a $k_0$ such that
\begin{equation}\label{eq:proof-deltakbound}
\norm{DL(t)P_{AB}^\perp}\leq C_1' \exp\left(- C_2 \sqrt{\lambda_k} \frac{l_k}{s_k} \right) ,\quad \forall k \ge k_0.
\end{equation}
From this fact we conclude the proof of Lemma \ref{lm:deltakbound}, by choosing $C_1$ to account for all the values of $k \le k_0$. 

 If $\frac{l_k}{s_k}$ is bounded and does not diverge, then the r.h.s. of \eqref{eq:proof-deltakbound} is not decreasing, and therefore it is trivially true by choosing $k_0=1$ and $C_1'$ appropriately. If $\frac{l_k}{s_k}$ is divergent, then for any $C_2 >0$ we can always choose $t = t(k)$ to be a the integer part of a multiple of $\frac{l_k}{s_k}$ such that there is a $k_0$ satisfying
\[
\frac{1}{\sqrt{1+g^2}} \left( \frac{t(k)}{C_\Gamma (L-1) R}-\frac{2L-1}{L-1}  \right) \ge C_2 , \quad \forall k \ge k_0.
\]
This implies that
\[
\norm{DL(t)P_{AB}^\perp}\leq 2 \exp\left(- C_2 \sqrt{\lambda_k} \frac{l_k}{s_k} \right), \quad \forall k \ge k_0.
\] 
\qed

\section*{Acknowledgments}
The authors thank the organizers of the online Lattice Seminar, during which the idea for this work arose and they thank Haruki Watanabe for useful discussions.
The research of M.L.\ is  supported by the DFG
through the grant TRR 352 – Project-ID 470903074 and by the European Union (ERC, MathQuantProp, project 101163620).\footnote{Views and opinions expressed are however those of the authors only and do not necessarily reflect those of the European Union or the European Research Council Executive Agency. Neither the European Union nor the granting authority can be held responsible for them.}
A.\,L.~acknowledges financial support from grants PID2020-113523GB-I00, PID2023-146758NB-I00, CEX2023-001347-S, funded by MICIU/AEI/10.13039/501100011033, grant RYC2019-026475-I funded by MICIU/AEI/10.13039/501100011033 and ``ESF Investing in your future'', as well as from ``Programma per Giovani Ricercatori Rita Levi Montalcini'' funded by by the Italian Ministry of University and Research (MUR).

\AtNextBibliography{\footnotesize}
    \printbibliography

\end{document}